\documentclass[review]{elsarticle}

\usepackage{lineno,hyperref}
\usepackage{amsmath,amsthm}

\newtheorem{theorem}{Theorem}
\newtheorem{lemma}{Lemma}
\newdefinition{remark}{Remark}
\newtheorem{corollary}{Corollary}

\journal{DMTCS}

\begin{document}

\begin{frontmatter}

\title{An upper bound for min-max angle of polygons}

\author[mymainaddress]{Saeed Asaeedi\corref{mycorrespondingauthor}}
\cortext[mycorrespondingauthor]{Corresponding author}
\ead{asaeedi@kashanu.ac.ir}

\author[mysecondaryaddress]{Farzad Didehvar}

\author[mysecondaryaddress]{Ali Mohades}

\address[mymainaddress]{Department of Computer Science, Faculty of Mathematical Sciences, University of Kashan, Kashan 87317-53153,I. R. Iran}
\address[mysecondaryaddress]{Department of Mathematics and Computer Science, Amirkabir University of Technology, Tehran, Iran}

\begin{abstract}
Let $S$ be a set of $n$ points in the plane, $\wp(S)$ be the set of all simple polygons crossing $S$, $\gamma_P$ be the maximum angle of polygon $P \in \wp(S)$ and $\theta =min_{P\in\wp(S)} \gamma_P$. In this paper, we prove that $\theta\leq 2\pi-\frac{2\pi}{r.m}$ where $m$ and $r$ are the number of edges and inner points of the convex hull of $S$, respectively. We also propose an algorithm to construct a polygon with the said upper bound on its angles. Constructing a simple polygon with angular constraint on a given set of points in the plane can be used for path planning in robotics. Moreover, we improve our upper bound on $\theta$ and prove that this is tight for $r=1$.
\end{abstract}


\end{frontmatter}


\section{Introduction}
An optimal polygonization of a set of points in the plane is a classical problem in computational geometry and has been applied to many fields such as image processing~\cite{marchand1999binary,pakhira2011digital}, pattern recognition~\cite{pakhira2011digital,pavlidis2013structural,abdi2009effective}, geographic information system~\cite{galton2006region}, etc. Considering a set $S$ of points in the plane, there are different numbers of simple polygons on $S$. Enumerating and generating simple polygons on $S$ has been the focus of many studies~\cite{zhu1996generating,nourollah2017use,garcia2000lower,wettstein2014counting,meijer1990upper}.

Finding polygons with special properties over all polygonizations is of particular interest to researchers. The minimum and maximum area polygonization are NP-complete, as shown by Fekete~\cite{fekete1993area,fekete2000simple}. The problems of computing the simple polygons with minimum and maximum perimeters is the well-known NP-complete problems called TSP and max-TSP, respectively. There are many ongoing studies on approximation algorithm for minimum and maximum area polygonization~\cite{taranilla2011approaching,peethambaran2016empirical}, TSP~\cite{bartal2016traveling,moylett2017quantum} and max-TSP~\cite{dudycz20174}.

In some of these approaches the angles have been investigated in many problems over polygonization. The Angular-Metric TSP~\cite{aggarwal2000angular} is the problem of finding a tour on $S$ minimizing the sum of the direction changes at each point. Fekete and Woeginger introduced Angle-Restricted Tour problem in~\cite{fekete1997angle}. For a set $A\subseteq (-\pi,\pi]$ of angles, Angle-Restricted Tour is the problem of finding a simple or non-simple polygon on $S$ where all angles of the polygon belong to $A$. In~\cite{asaeedi2017alpha} $\alpha$-concave hull refers to a simple polygon $P$ with minimum area covering a set of points such that all angels of $P$ are less than or equal to $\pi+\alpha$.

Reflexivity, the smallest number of reflex vertices among all polygonizations of a set of points, is considered as a convexity measurement for those points. Arkin et al.~\cite{arkin2003reflexivity} introduced the concept of reflexivity and presented lower and upper bounds for reflexivity of any set of n points. E. Ackerman et al.~\cite{ackerman2009improved} improved the upper bound and proposed an algorithm to compute polygon with at most this number of reflex vertices in the time complexity of $O(n \log n)$. In~\cite{lien2008approximate} a convexity measurement has been proposed for polyhedra.

Rorabaugh~\cite{citation-0} investigated the min-max value of reflex angles in polygonizations as another convexity measurement for a set of points and derived an upper bound for their solution.

In~\cite{asaeedi2019nlp}, the upper bound $2\pi-\frac{2\pi}{2^{r-1}.m}$ is presented for min-max value of the angles in polygonization where $m$ and $r$ are the number of edges and inner points of the convex hull of $S$, respectively. Here we improved this upper bound to $2\pi-\frac{2\pi}{r.m}$.

The rest of the paper is as follows: In the section 2, notations and definitions are presented. In section 3, the upper bound is derived and in section 4, we conclude the paper highlighting its achievements.

\section{Preliminaries}

Let $S=\{s_1,s_2,...,s_n \}$ be a set of points in the plane and $CH$ be the convex hull of $S$.
The vertices and edges of $CH$ are denoted by $V_{CH}=\{c_1,c_2,...,c_m \}$ and $E_{CH}=\{e_1,e_2,...,e_m\}$, respectively. Furthermore, let $IP=\{a_1,a_2,...,a_r\}$ be the inner points of $CH$ where $r=n-m$. Table~\ref{tab:1} shows more notations that are used in the rest of the paper. A polygon $P$ crossing $S$ is specified by a closed chain of vertices $P=(p_1, p_2,...,p_n, p_1)$ such that $S=V_P=\{p_1, p_2,...,p_n\}$.

\begin{table}
\caption{Notations of symbols}
\label{tab:1}  
\resizebox{\textwidth}{!}{
\begin{tabular}{lll}
\hline\noalign{\smallskip}
Notation & Description  \\
\noalign{\smallskip}\hline\noalign{\smallskip}
$S$ & A set of points in the plane \\
$n$ & cardinality of $S$  \\
$s_i$ & $i$th point of $S$ ($1\leq i\leq n$)  \\
$CH$ & convex hull of $S$  \\
$m$ & number of vertices of $CH$ \\
$IP$ & inner points of $CH$ \\
$r$ & cardinality of $IP$ \\
$P$ & a simple Polygon crossing $S$ \\
$V_P$ & vertices of $P$ \\
$E_P$ & edges of $P$ \\
$c_j$ & $j$th vertex of $CH$ ($1\leq j\leq m$) \\
$e_j$ & $j$th edge of $CH$ ($1\leq j\leq m$) \\
$\overline{s_i s_j}$ & an edge of $P$ with $s_i$ and $s_j$ as its end points ($1\leq i,j\leq n$, $i\neq j$)\\
$\wp(S)$	& set of all simple polygons crossing $S$ \\
$\alpha$, $\beta$, $\gamma$, $\theta$ & angles between 0 and $2\pi$ \\
\noalign{\smallskip}\hline
\end{tabular}
}
\end{table}

Let $e=\overline{AB}$ be a line segment. The minor arc $\stackrel{\frown}{AB}$ with measure equal to $\alpha$ is denoted by $s_e^{\alpha}$, and the major arc $\stackrel{\frown}{AB}$ with measure equal to $\beta=2\pi-\alpha$ is denoted by $S_e^{\beta}$. We denote the minor and major segments on $e$ by $m_e^\alpha$ and $M_e^\beta$, respectively (see Fig.~\ref{fig:1}). Also, we use the concept of "Sweep Arc" in our algorithm which is defined in~\cite[Section 4]{asaeedi2019nlp}. A sweep arc on $e$ is a minor arc $s_e^0$ where it expands to the major arc $S_e^{2\pi}$.

\begin{figure}[h]
\centering
  \includegraphics[width=0.25\textwidth]{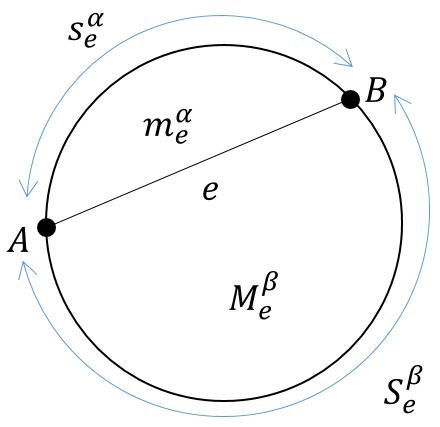}
\caption{The notations of minor arc, major arc, minor segment and major segment on $e$}
\label{fig:1} 
\end{figure}

\section{Min-Max Angle}
In this section we present two upper bounds for $\theta$. Let us first present a lemma followed by a theorem.

\begin{lemma}
\label{lem:3}
Let $l =\overline{c_1 c_2}$ be a line segment and $S$ be a set of $n$ points inside the $M_l^{\beta_{max}}$, where $\beta_{max}=2\pi-\frac{4\pi}{m}$ for an integer number $m$. Assume that $t$ points $\{s_1,s_2,...,s_t \}$ are met by the sweep arc on $l$ and $P=(c_1,s_1,s_2,...,s_t,c_2,c_1)$ is a simple polygon such that all internal angles of $\hat{s_i}$ are greater than or equal to $\frac{2\pi}{t.m}$. Let $x$ be $(t+1)$th point met by the sweep arc. There exists an edge $\overline{ab}$ of $P$ such that $\widehat{axb}$ is greater than or equal to $\frac{2\pi}{(t+1).m}$.
\end{lemma}

\begin{proof}
We prove the lemma by induction on $t$. When $t=0$, the chain $P=(c_1,s_1,s_2,...,s_t,c_2,c_1)$ is a line segment $\overline{c_1 c_2 }$. Therefore, we consider both cases $t=0$ and $t=1$ as the base cases.
\paragraph{Base case $(t=0)$} Let $x$ be the first point that the sweeping arc meets. We construct the polygon by connecting $x$ to $c_1$ and $c_2$. Since the maximum measure of the arc is $\beta_{max}$, the internal angle of $\gamma=\widehat{c_1 xc_2}$ in the triangle $\triangle{c_1 xc_2}$ is greater than or equal to $\frac{2\pi}{m}$ (see Fig.~\ref{fig:7}).

\begin{figure}[h]
\centering
  \includegraphics[width=0.35\textwidth]{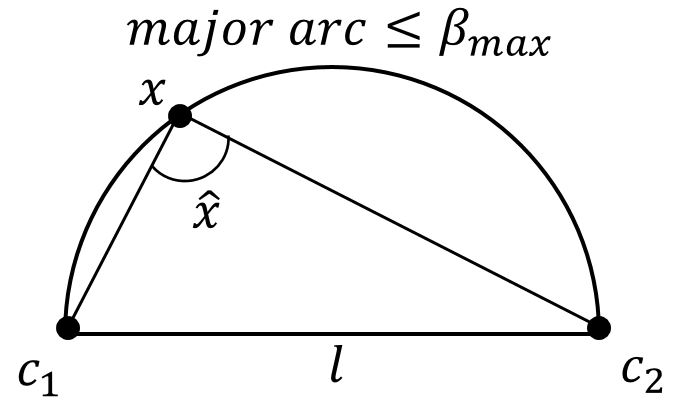}
\caption{$\hat{x}$ is greater than $\frac{2\pi}{m}$.}
\label{fig:7} 
\end{figure}

\paragraph{Base case $(t=1)$} Let $s_1$ be the first point that the sweeping arc meets and $x$ be the second one. Also, let $e_1=\overline{c_1 s_1}$ and $e_2=\overline{s_1 c_2}$ be two edges of $P=(c_1,s_1,c_2,c_1)$. The edges $e_1$ and $e_2$ divide the sweeping arc into 3 parts; the arc $B_1$ where $e_1$ is visible but $e_2$ is not visible from all the points on it, the arc $B_2$ where $e_2$ is visible but $e_1$ is not visible from all the points on it, and finally the arc $B_3$ where $e_1$ and $e_2$ are both visible from all the points on it (see Fig.~\ref{fig:8}).

\begin{figure}[h]
  \includegraphics[width=0.95\textwidth]{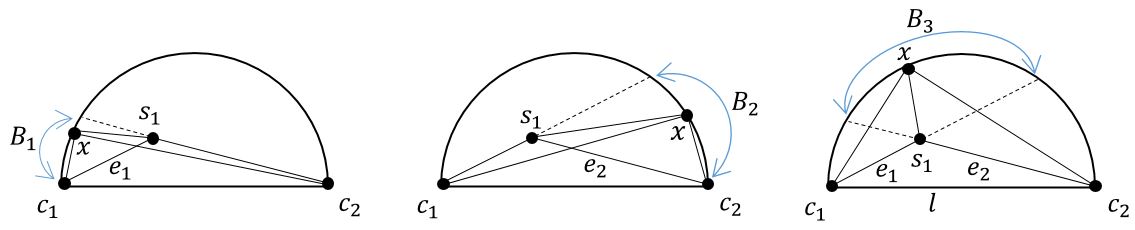}
\caption{The edges $e_1$ and $e_2$ divide the sweeping arc into 3 parts: $B_1$, $B_2$ and $B_3$.}
\label{fig:8} 
\end{figure}

\begin{enumerate}[{case} 1.]
\item 
If $x$ is placed on $B_1$, the angle $\widehat{c_1 xs_1}$ is greater than $\gamma=\widehat{c_1 xc_2}$ and the angle $\gamma$ is greater than or equal to $\frac{2\pi}{m}$. Hence, the angle $\widehat{c_1 xs_1}$ is greater than $\frac{2\pi}{m}$. So, we consider the edge $\overline{c_1 s_1 }$ as the desired edge $\overline{ab}$ such that $\widehat{axb}$ is greater than or equal to $\frac{\pi}{m}$.
\item
If $x$ is placed on $B_2$, the angle $\widehat{s_1 xc_2}$ is greater than $\gamma$ and the angle $\gamma$ is greater than or equal to $\frac{2\pi}{m}$. Hence, the angle $\widehat{s_1 xc_2}$ is greater than $\frac{2\pi}{m}$. So, we consider the edge $\overline{s_1 c_2}$ as the desired edge $\overline{ab}$ such that $\widehat{axb}$ is greater than or equal to $\frac{\pi}{m}$.
\item
If $x$ is placed on $B_3$, the maximum of $\widehat{c_1 xs_1}$ and $\widehat{s_1 xc_2}$ is greater than $\frac{\gamma}{2}$. Since $\gamma$ is greater than $\frac{2\pi}{m}$, the maximum of $\widehat{c_1 xs_1}$ and $\widehat{s_1 xc_2}$ is greater than $\frac{2\pi}{2m}$. Hence, if $\widehat{c_1 xs_1}$ is greater than $\widehat{s_1 xc_2}$, the edge $\overline{c_1 s_1}$ is considered as $\overline{ab}$, otherwise, the edge $\overline{s_1 c_2}$ is considered as $\overline{ab}$.
\end{enumerate}

\paragraph{Induction assumption} Let $y$ be $k$th point that the sweeping arc meets. There exists an edge $\overline{ab}$ of $P=(c_1,s_1,s_2,...,s_{k-1},c_2,c_1)$ such that $\widehat{ayb}$ is greater than or equal to $\frac{2\pi}{k.m}$.

\paragraph{Inductive step} Let $x$ be $(k+1)$th point that the sweeping arc meets and $P=(c_1,s_1,s_2,…,s_k,c_2,c_1)$ be the polygon such that all internal angles of $\hat{s_i}$ are greater than or equal to $\frac{2\pi}{k.m}$. We show that there exists an edge $\overline{ab}$ of $P$ such that $\widehat{ayb}$ is greater than or equal to $\frac{2\pi}{(k+1).m}$. Here, three cases need to be examined:

\begin{enumerate}[{case} 1.]
\item
All edges of $P$ except $\overline{c_1 c_2}$ are visible from $x$. Let $e_1=\overline{c_1 s_1}$, $e_2=\overline{s_1 s_2}$, ... ,and $e_{k+1}=\overline{s_k c_2}$ be edges of $P$, $\beta_i$ be the angle subtended by $e_i$ at the point $x$ and $\beta_M$ be the maximum one. Since the angle $\gamma$ is greater than or equal to $\frac{2\pi}{m}$ and $\Sigma_{i=1}^{k+1} \beta_i =\gamma$, we have $\beta_M>\frac{2\pi}{(k+1).m}$. Let $e$ be the edge that corresponds to $\beta_M$. So, the edge $e$ is considered as $\overline{ab}$ such that $\widehat{axb}$ is greater than or equal to $\frac{2\pi}{(k+1).m}$ (see Fig.~\ref{fig:9}).

\begin{figure}[h]
\centering
  \includegraphics[width=0.45\textwidth]{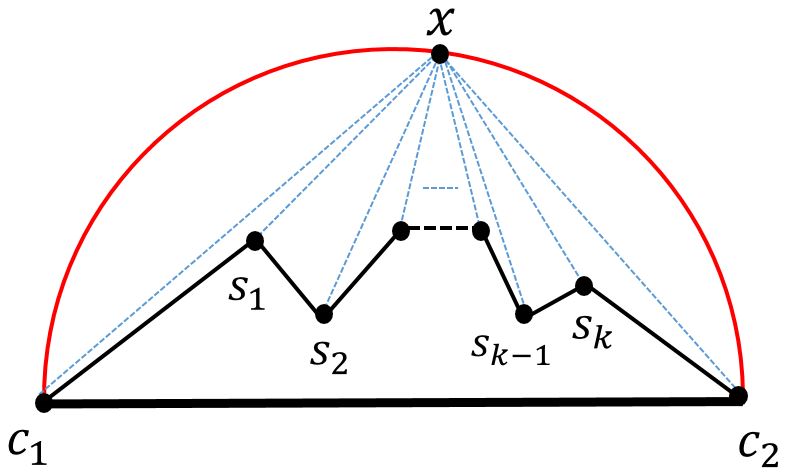}
\caption{All edges $e_i$ are visible from $x$.}
\label{fig:9} 
\end{figure}

\item
There exists an edge $e=\overline{cd}$ of $P$ such that both endpoints of $e$ are not visible from $x$. We obtain a polygon $P'$ from $P$ by contracting~\cite{harary1969graph} the edge $e$, i.e. $P'=P/e$. Since $P'$ has $k+1$ vertex points, by induction assumption, there exists an edge $e'=\overline{ab}$ of $P'$ such that $\widehat{axb}$ is greater than or equal to $\frac{2\pi}{k.m}$. The polygon $P''$ is obtained from $P$ by removing the edge $\overline{ab}$ and adding two edges $\overline{ax}$ and $\overline{xb}$. Since the two end points of $e=\overline{cd}$ are invisible from $x$, contracting and splitting $e$ has no effect on the measure of the angle $\widehat{axb}$. Hence, the angle $\widehat{axb}$ in $P''$ is greater than or equal to $\frac{2\pi}{(k+1).m}$ (see Fig.~\ref{fig:10}).

\begin{figure}[h]
  \includegraphics[width=1\textwidth]{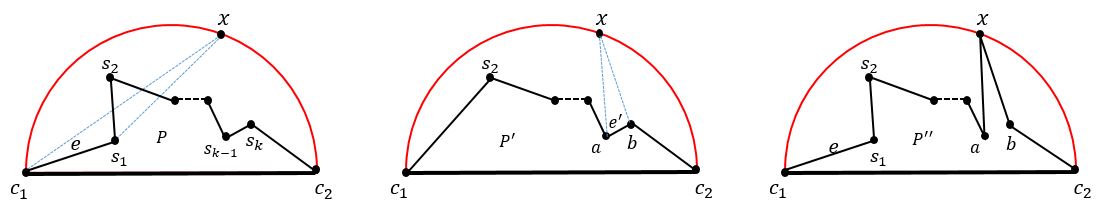}
\caption{The edge $e$ of $P$ is invisible from $x$. Contracting $e$ leads to construct $P'$ from $P$. The polygon $P''$ obtained from $P$ and $P'$.}
\label{fig:10} 
\end{figure}

\item
There exists an edge $e=\overline{cd}$ of $P$ such that one endpoint of $e$ is not visible from $x$ (see Fig.~\ref{fig:11}). We obtain a polygon $P'=P/e$ from $P$ by contracting the edge $e$. Since $P'$ has $k+1$ vertex points, by induction assumption, there exists an edge $e'=\overline{ab}$ of $P'$ such that $\widehat{axb}$ is greater than or equal to $\frac{2\pi}{k.m}$. The polygon $P''$ is obtained from $P$ by removing the edge $\overline{ab}$ and adding two edges $\overline{ax}$ and $\overline{xb}$.

\begin{figure}[h]
\centering
  \includegraphics[width=0.45\textwidth]{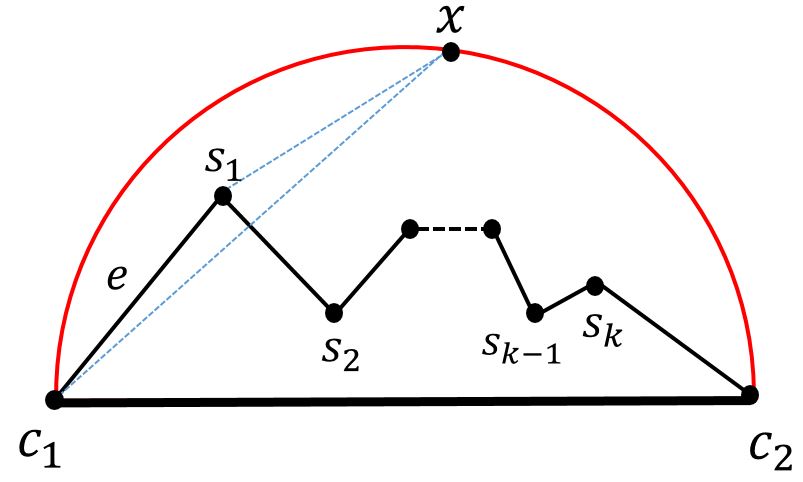}
\caption{The vertex $s_1$ is visible and $c_1$ is invisible from $x$.}
\label{fig:11} 
\end{figure}

If either $a$ or $b$ in $P''$ be an endpoint of $e$, contracting and splitting $e$ has an effect on the measure of the angle $\widehat{axb}$ (see fig.~\ref{fig:12}). In other words, the angle $\widehat{axb}$ in $P''$ is not equal to the angle $\widehat{axb}$ in $P'$. It is clear that the angle $\widehat{axb}$ in $P''$ is greater than the angle $\widehat{axb}$ in $P'$. Since the angle $\widehat{axb}$ in $P'$ is greater than or equal to $\frac{2\pi}{k.m}$, the angle $\widehat{axb}$ in $P''$ is greater than or equal to $\frac{2\pi}{(k+1).m}$.

\begin{figure}[h]
\centering
  \includegraphics[width=0.75\textwidth]{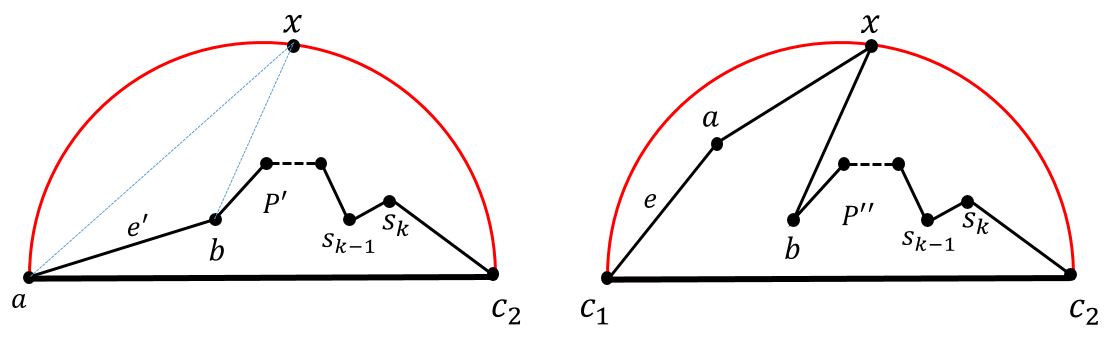}
\caption{$\widehat{axb}$ in polygon $P''$ is greater than $\widehat{axb}$ in polygon $P'$.}
\label{fig:12} 
\end{figure}

 Also, if both points $a$ or $b$ in $P''$ are not the endpoints of $e$, contracting and splitting $e$ has no effect on the measure of the angle $\widehat{axb}$ (see fig.~\ref{fig:13}). In other words, the angle $\widehat{axb}$ in $P''$ is equal to the angle $\widehat{axb}$ in $P'$. Hence, the angle $\widehat{axb}$ in $P''$ is greater than or equal to $\frac{2\pi}{(k+1).m}$.
\end{enumerate}

\begin{figure}[h]
\centering
  \includegraphics[width=0.75\textwidth]{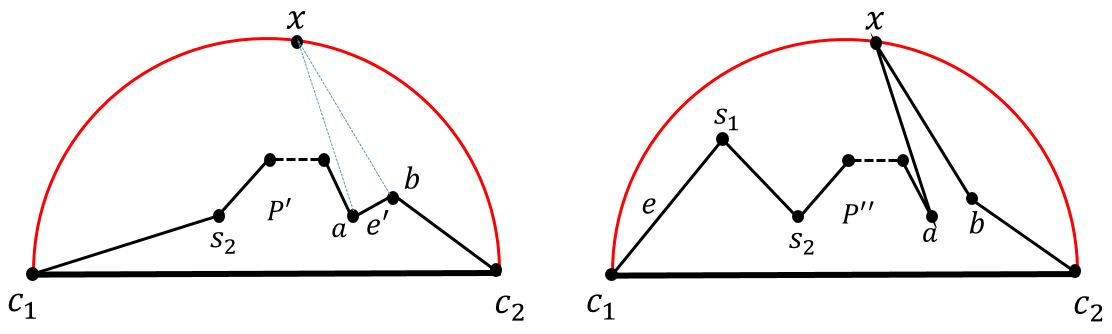}
\caption{$\widehat{axb}$ in polygon $P''$ is equal to $\widehat{axb}$ in polygon $P'$.}
\label{fig:13} 
\end{figure}
\end{proof}

\begin{theorem}
\label{lem:4}
Let $l =\overline{c_1 c_2}$ be a line segment and $S$ be a set of $n$ points inside the $M_l^{\beta_{max}}$, such that $\beta_{max}=2\pi-\frac{4\pi}{m}$ for an integer number $m$ (see Fig.~\ref{fig:14}.a). There exists a chain $(s_1,s_2,...,s_n)$ on $S$ such that all internal angles of $\hat{s_i}$ in the polygon $(c_1,s_1,s_2,...,s_n,c_2,c_1)$ are greater than or equal to $\frac{2\pi}{n.m}$ (see Fig.~\ref{fig:14}.b).
\end{theorem}

\begin{figure}[h]
\centering
  \includegraphics[width=0.65\textwidth]{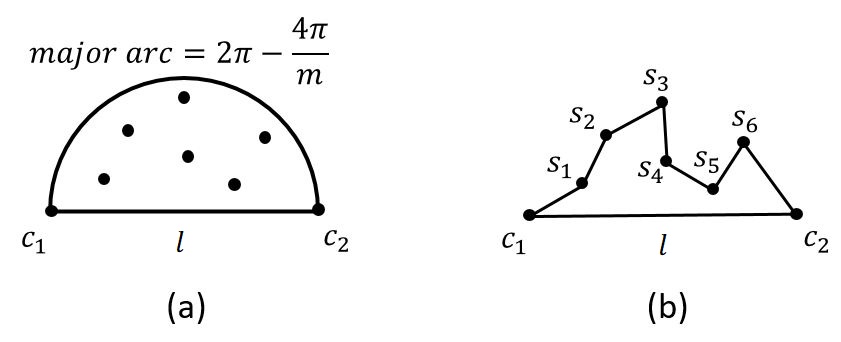}
\caption{a. $S$ is a set of 6 points inside $M_l^{\beta_{max}}$. \qquad b. $\forall 1\leq i\leq 6$, $\hat{s_i}\geq \frac{2\pi}{6m}$ .}
\label{fig:14} 
\end{figure}

\begin{proof}
We prove theorem~\ref{lem:4} by constructing the polygon $(c_1,s_1,s_2,...,s_n,c_2,c_1)$, using the following algorithm which is a modified version of that originally presented in~\cite{asaeedi2019nlp}:
\paragraph{\textbf{Algorithm 1 (Modified Sweep Arc Algorithm)}} 
\begin{enumerate}
\item
Sweep the arc $\stackrel{\frown}{c_1 c_2}$ from $s_l^0$ to $S_l^{\beta_{max}}$.
\item
Let $x_1$ be the first point which is met by the sweep arc. Construct $P=(c_1,x_1,c_2,c_1)$ as the desired polygon.
\item
Set $i=2$.
\item
Let $P=(c_1,s_1,s_2...,s_{i-1},c_2,c_1)$ be the constructed polygon inside the sweep arc and $x_i$ be the $i$th point which is met by the sweep arc.
\item
Assume that $e_1=\overline{c_1 s_1}$, $e_2=\overline{s_1 s_2}$, ... , and $e_i=\overline{s_{i-1} c_2}$ are the edges of $P$. If $e_j$ is visible from $x$, set $\beta_j=$The angle subtended by $e_j$ at the point $x$, otherwise set $\beta_j=0$.
\item
Let $\beta_M=\max_{1 \leq j \leq i} \beta_j$ and $e=\overline{ab}$ be the edge that corresponds to $\beta_M$.
\item
Remove the edge $e$ from $P$ and add two edges $\overline{ax_i}$ and $\overline{x_i b}$ to construct the desired polygon.
\item
Set $i=i+1$. If $i\leq n$, then go to 4, otherwise exit.
\end{enumerate}
Based on Lemma~\ref{lem:3}, $\forall j\in\{1,2,...,i\}$ the angles $\hat{s_j}$ in $P$ are greater than or equal to $\frac{2\pi}{i.m}$ in step 4 of the algorithm. Therefore, when $i=n$, the angles $\hat{s_j}$ in $P$ are greater than or equal to $\frac{2\pi}{n.m}$.
\end{proof}


It is proved in~\cite[Lemma 3]{asaeedi2019nlp} that all angles of the mentioned polygon $P$ are greater than or equal to $\frac{2\pi}{2^{n-1}.m}$. Here, based on theorem~\ref{lem:4}, we increase this bound to $\frac{2\pi}{n.m}$. This yields us the following corollaries:


\begin{corollary}
\label{cor:1}
Let $S$ be a set of points in the plane, $CH$ be the convex hull of $S$ and $m$ and $r$ be the number of edges and inner points of $CH$, respectively. 
If we replace algorithm 1 of~\cite[Theorem 2, Step 2.a of Algorithm 2]{asaeedi2019nlp} by the modified sweep arc algorithm, the upper bound $2\pi-\frac{2\pi}{r.m}$ is achieved for $\theta$.
\end{corollary}

\begin{remark}
\label{rmk:1}
 Based on corollary~\ref{cor:1}, in the case of $r=1$, $2\pi-\frac{2\pi}{n-1}$ is an upper bound for $\theta$ over all simple polygons crossing $S$. It is noteworthy that this bound is tight in this case. The tightness is achieved when the inner point is at the center of a regular n-gons, as illustrated in Fig.~\ref{fig:5}.
\end{remark}

\begin{figure}[h]
\centering
  \includegraphics[width=0.3\textwidth]{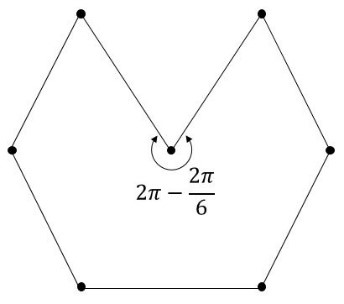}
\caption{Maximum angle of each polygon crossing these points is equal to $2\pi-\frac{2\pi}{6}$}
\label{fig:5} 
\end{figure}


The following corollary improved the upper bound to $2\pi-\frac{2\pi}{d.m}$ where $d$ is depth of angular onion peeling on $S$ which is defined in~\cite{asaeedi2019nlp}.

\begin{corollary}
\label{cor:2}
Let $S$ be a set of points in the plane, $CH$ be the convex hull of $S$, $m$ be the cardinality of edges of $CH$ and $d$ be the depth of angular onion peeling on $S$. 
If we replace algorithm 1 of~\cite[Theorem 3, Step 2.b of Algorithm 3]{asaeedi2019nlp} by the modified sweep arc algorithm, the upper bound $2\pi-\frac{2\pi}{d.m}$ is achieved for $\theta$.
\end{corollary}

Since the time complexity of modified sweep arc algorithm is $O(r)$, those of both modified algorithm 2 and 3 are $O(n\log{n}+rm)$. Note that the modified algorithm 2 and 3 are those proposed in~\cite{asaeedi2019nlp} in which the algorithm 1 is replaced by the modified sweep arc algorithm. Based on corollary~\ref{cor:1}, the modified algorithm 2 constructs a polygon such that its internal angles are less than or equal to $2\pi-\frac{2\pi}{r.m}$. Based on corollary~\ref{cor:2}, this bound is improved to $2\pi-\frac{2\pi}{d.m}$ using modified algorithm 3. When $S$ is a set of $n$ points in the plane and the convex hull of $S$ has $n-1$ edges, the depth of angular onion peeling on $S$ is equal to 1. Hence, the upper bound for $\theta$ is equal to $2\pi-\frac{2\pi}{1.(n-1)}$ which confirms the remark~\ref{rmk:1}.

Computing $\alpha$-concave hull on a set $S$ of points is an NP-complete problem~\cite{asaeedi2017alpha}. For all $\alpha>\theta$, $\alpha$-concave hull crosses all points of $S$. So, the polygon computed by modified algorithm 3 is an $\alpha$-polygon~\cite{asaeedi2017alpha} which approximates $\alpha$-concave hull of $S$. The following corollary shows the relation between $\alpha$-concave hull and the computed upper bound.

\begin{corollary}
Let $S$ be a set of points in the plane, $CH$ be the convex hull of $S$, $m$ be the cardinality of edges of $CH$ and $d$ be the depth of angular onion peeling on $S$. For all $\alpha>2\pi-\frac{2\pi}{d.m}$, there always exists an $\alpha$-concave hull $P$ on $S$ such that $P$ crosses all points of $S$.
\end{corollary}

Coverage path planning is a fundamental problem in the field of robotics. There are many limitation factors in order to plan a path for a robot to cover (or visit) all points of a set of points, such as robot rotation angle. The following corollary presents the essential relation between path planning in robotics and our upper bounds on $\theta$. 
\begin{corollary}
Let $S$ be a set of $n$ points in the plane, $CH$ be the convex hull of $S$, $m$ be the cardinality of edges of $CH$ and $d$ be the depth of angular onion peeling on $S$. If the robot rotation angle is greater than $2\pi-\frac{2\pi}{d.m}$, there always exists a path for the robot to cover $S$. As stated before, this path can be found in $O(n\log{n}+rm)$.
\end{corollary}

\section{Conclusion}
The major problem investigated in this paper is that of finding a simple polygon with angular constraint on a given set of points in the plane. We derived the upper bounds for min-max value of angles over all simple polygons crossing the given set of points. We also presented algorithms to compute the polygons thereby satisfying the derived upper bounds. In addition to the theoretical results, this bound is an important achievement in the field of robotic.

\section{Acknowledgements}
The research of the first author is partially supported by the University of Kashan under grant number 991449/1.


\bibliography{mybibfile}

\begin{thebibliography}{10}
\expandafter\ifx\csname url\endcsname\relax
  \def\url#1{\texttt{#1}}\fi
\expandafter\ifx\csname urlprefix\endcsname\relax\def\urlprefix{URL }\fi
\expandafter\ifx\csname href\endcsname\relax
  \def\href#1#2{#2} \def\path#1{#1}\fi

\bibitem{marchand1999binary}
S.~Marchand-Maillet, Y.~M. Sharaiha, Binary digital image processing: a
  discrete approach, Academic Press, 1999.

\bibitem{pakhira2011digital}
M.~K. Pakhira, Digital image processing and pattern recognition, PHI Learning
  Pvt. Limited, 2011.

\bibitem{pavlidis2013structural}
T.~Pavlidis, Structural pattern recognition, Vol.~1, Springer, 2013.

\bibitem{abdi2009effective}
M.~N. Abdi, M.~Khemakhem, H.~Ben-Abdallah, An effective combination of mpp
  contour-based features for off-line text-independent arabic writer
  identification, in: Signal processing, image processing and pattern
  recognition, Springer, 2009, pp. 209--220.

\bibitem{galton2006region}
A.~Galton, M.~Duckham, What is the region occupied by a set of points?,
  Geographic Information Science (2006) 81--98.

\bibitem{zhu1996generating}
C.~Zhu, G.~Sundaram, J.~Snoeyink, J.~S. Mitchell, Generating random polygons
  with given vertices, Comput. Geom. 6~(5) (1996) 277--290.

\bibitem{nourollah2017use}
A.~Nourollah, M.~Movahedinejad, Use of simple polygonal chains in generating
  random simple polygons, Japan Journal of Industrial and Applied Mathematics
  34~(2) (2017) 407--428.

\bibitem{garcia2000lower}
A.~Garc{\i}a, M.~Noy, J.~Tejel, Lower bounds on the number of crossing-free
  subgraphs of kn, Computational Geometry 16~(4) (2000) 211--221.

\bibitem{wettstein2014counting}
M.~Wettstein, Counting and enumerating crossing-free geometric graphs, in:
  Proceedings of the thirtieth annual symposium on Computational geometry, ACM,
  2014, p.~1.

\bibitem{meijer1990upper}
H.~Meijer, Upper and lower bounds for the number of monotone crossing free
  hamiltonian cycles from a set of points, Ars Combinatoria 30 (1990) 203--208.

\bibitem{fekete1993area}
S.~P. Fekete, W.~R. Pulleyblank, Area optimization of simple polygons, in:
  Proceedings of the ninth annual symposium on Computational geometry, ACM,
  1993, pp. 173--182.

\bibitem{fekete2000simple}
S.~P. Fekete, On simple polygonalizations with optimal area, Discrete \&
  Computational Geometry 23~(1) (2000) 73--110.

\bibitem{taranilla2011approaching}
M.~T. Taranilla, E.~O. Gagliardi, G.~Hern{\'a}ndez~Pe{\~n}alver, Approaching
  minimum area polygonization.

\bibitem{peethambaran2016empirical}
J.~Peethambaran, A.~D. Parakkat, R.~Muthuganapathy, An empirical study on
  randomized optimal area polygonization of planar point sets, Journal of
  Experimental Algorithmics (JEA) 21~(1) (2016) 1--10.

\bibitem{bartal2016traveling}
Y.~Bartal, L.-A. Gottlieb, R.~Krauthgamer, The traveling salesman problem:
  low-dimensionality implies a polynomial time approximation scheme, SIAM
  Journal on Computing 45~(4) (2016) 1563--1581.

\bibitem{moylett2017quantum}
D.~J. Moylett, N.~Linden, A.~Montanaro, Quantum speedup of the
  traveling-salesman problem for bounded-degree graphs, Physical Review A
  95~(3) (2017) 032323.

\bibitem{dudycz20174}
S.~Dudycz, J.~Marcinkowski, K.~Paluch, B.~Rybicki, A 4/5-approximation
  algorithm for the maximum traveling salesman problem, in: International
  Conference on Integer Programming and Combinatorial Optimization, Springer,
  2017, pp. 173--185.

\bibitem{aggarwal2000angular}
A.~Aggarwal, D.~Coppersmith, S.~Khanna, R.~Motwani, B.~Schieber, The
  angular-metric traveling salesman problem, SIAM Journal on Computing 29~(3)
  (2000) 697--711.

\bibitem{fekete1997angle}
S.~P. Fekete, G.~J. Woeginger, Angle-restricted tours in the plane,
  Computational Geometry 8~(4) (1997) 195--218.

\bibitem{asaeedi2017alpha}
S.~Asaeedi, F.~Didehvar, A.~Mohades, $\alpha$-concave hull, a generalization of
  convex hull, Theoretical Computer Science 702 (2017) 48--59.

\bibitem{arkin2003reflexivity}
E.~M. Arkin, J.~S. Mitchell, S.~P. Fekete, F.~Hurtado, M.~Noy,
  V.~Sacrist{\'a}n, S.~Sethia, On the reflexivity of point sets, in: Discrete
  and Computational Geometry, Springer, 2003, pp. 139--156.

\bibitem{ackerman2009improved}
E.~Ackerman, O.~Aichholzer, B.~Keszegh, Improved upper bounds on the
  reflexivity of point sets, Computational geometry 42~(3) (2009) 241--249.

\bibitem{lien2008approximate}
J.-M. Lien, N.~M. Amato, Approximate convex decomposition of polyhedra and its
  applications, Computer Aided Geometric Design 25~(7) (2008) 503--522.

\bibitem{citation-0}
J.~O’Rourke, S.~Suri, C.~D. T{\'o}th, 30 polygons.

\bibitem{asaeedi2019nlp}
S.~Asaeedi, F.~Didehvar, A.~Mohades, Nlp formulation for polygon optimization
  problems, Mathematics 7~(1) (2019) 24.

\bibitem{harary1969graph}
F.~Harary, Graph theory, Tech. rep., MICHIGAN UNIV ANN ARBOR DEPT OF
  MATHEMATICS (1969).

\end{thebibliography}

\end{document}